\documentclass[11pt]{article}

\usepackage{amsmath,amssymb,mathtools,amsthm}
\usepackage{geometry}
\usepackage{etoolbox}
\usepackage{tikz}
\usepackage[colorlinks=true,citecolor=blue,linkcolor=blue,urlcolor=blue]{hyperref}
\usetikzlibrary{positioning}
\newcommand{\matzero}{\textcolor{black!35}{0}}
\newcommand{\hxedgeone}{\colorbox{cyan!22}{\(\textcolor{cyan!70!black}{\mathbf{1}}\)}}
\newcommand{\hzedgeone}{\colorbox{red!18}{\(\textcolor{red!70!black}{\mathbf{1}}\)}}

\newtheoremstyle{uprightthm}
  {6pt}
  {6pt}
  {\normalfont}
  {}
  {\bfseries}
  {.}
  {0.5em}
  {}
\theoremstyle{uprightthm}
\newtheorem{theorem}{Theorem}
\newtheorem{definition}{Definition}
\AtEndEnvironment{theorem}{\hfill$\square$}
\AtEndEnvironment{definition}{\hfill$\square$}

\title{A Factor-Graph Formulation of CSS Syndrome Decoding:\\
Joint BP and Four-State BP}
\author{Kenta Kasai\\
Institute of Science Tokyo\\
\texttt{kenta@ict.eng.isct.ac.jp}}
\date{}

\begin{document}
\maketitle

\begin{abstract}
For CSS syndrome decoding, the two check matrices impose binary parity-check
constraints on the two Pauli error components.  The posterior can therefore be
written as a binary factor graph with two Tanner graphs coupled by the local
joint prior at each qubit.  While a four-state Pauli-label representation is
natural for general stabilizer-code decoding, the binary CSS factor graph is the
natural representation for CSS codes.  We call the sum-product algorithm on this
factorization joint belief propagation (joint BP).  Joint BP retains the local
channel correlation between the two Pauli
components.  This note compares joint BP with the four-state Pauli-label factor
graph used for four-state BP.
The two algorithms are shown to have the same posterior weights, messages, and
beliefs after relabeling the four local Pauli states and marginalizing the
irrelevant binary component.
\end{abstract}

\section{Introduction}

Classical low-density parity-check (LDPC) codes were introduced as sparse
binary parity-check codes \cite{gallager62} and later formalized as graph codes
\cite{tanner81}.  Belief propagation (BP) for these codes is the
sum-product algorithm applied to a factorization of a probability distribution
\cite{ajiMcEliece00,kschischang01}.  This factor-graph formulation is the
standard language behind density evolution, irregular ensembles, and other
binary-LDPC design tools \cite{richardson01,chung01,luby01,ru2008}.

CSS codes are constructed from two binary linear codes with orthogonal check
matrices \cite{calderbank96,steane96}.  In syndrome
decoding, one binary check matrix constrains the \(Z\)-component of the Pauli
error and the other constrains the \(X\)-component.  Therefore the CSS
posterior has a binary factorization: two binary Tanner graphs are connected
only through the local channel prior \(Q_j(x_j,z_j)\) at each qubit.  If this
joint prior is replaced by the two marginal priors, the posterior factorization
separates into two independent binary decoders.  That decoder is the
separate belief-propagation decoder, abbreviated separate BP below, used or
discussed in several quantum LDPC decoding papers
\cite{babar15,rigbyOlivierJarvis19,roffe20,panteleevKalachev21,ostrev24,
huangUeng26}.  For a Pauli channel with correlated \(X\)- and \(Z\)-components,
the separated factorization does not contain the local correlation.

Sparse quantum-code decoding has also been formulated with one four-state Pauli
label per qubit, equivalently with labels in \(\mathbb F_4\) for stabilizer-code
BP \cite{mackay04,poulin08,rigbyOlivierJarvis19,kuo20}.  This representation
keeps the four local Pauli probabilities in a single variable-node prior.
Several works address the loss of \(X/Z\) correlation in separated CSS decoding
by adding or preserving correlation information, including CSS decoders using
\(X/Z\) correlations \cite{delfosseTillich14}, modified \(\mathbb F_2\)-based BP
decoders that reintroduce error correlations \cite{rigbyOlivierJarvis19},
soft-decision CSS decoding based on joint probabilities of two binary components
\cite{gotoUchikawa13,gotoUchikawa14}, and binary factor-graph sum-product
decoding with local joint priors \cite{komotoKasai25}.  Related work also shows
that refined BP can reproduce the output of conventional quaternary BP with
binary-BP-like check-node complexity \cite{kuo20}.

This note isolates a narrower statement.  It does not introduce another
decoder and does not claim that binary BP always means separate BP.  We call
the coupled binary decoder \emph{joint BP}: it is the sum-product algorithm on
the binary CSS factor graph that keeps the joint local prior
\(Q_j(x_j,z_j)\).  The note starts from two posterior factorizations for CSS
syndrome decoding: the joint-BP factorization with factors \(Q_j(x_j,z_j)\),
and the four-state Pauli-label factorization obtained by the relabeling
\(\alpha=x+\omega z\).  It then writes the sum-product updates for both graphs
and proves that the messages and beliefs correspond by relabeling and
marginalization.  Thus the ability to use local \(X/Z\) channel correlation is
a property of the posterior factorization, not of whether the sum-product
recursion is represented with binary messages or four-state Pauli labels.

The scope is CSS syndrome decoding for qubits.  The \(\mathbb F_4\) notation in
this note is a label set for the four Pauli error states.  It should be
distinguished from code-construction theories based on additive or linear
codes over \(\mathbb F_4\) \cite{calderbankRainsShorSloane98}, from
generalizations to \(p^m\)-state systems \cite{matsumotoUyematsu00}, and from
non-binary component-code constructions over larger finite fields
\cite{kasaiHagiwaraImaiSakaniwa12}.
For non-CSS stabilizer codes, a four-state Pauli-label decoder is often the
natural factor-graph representation because the checks generally couple the
local Pauli labels directly.  For CSS codes, the syndrome constraints are
already binary parity checks on the two components, so the joint-BP
factorization is the natural starting point.

\section{Factor Graphs and Sum-Product}

We use belief propagation in its standard meaning: the sum-product algorithm
applied to a probability factorization.  Consider variables
\(\boldsymbol{u}=(u_j)_{j\in V}\), where \(u_j\) takes values in a finite
alphabet \(\mathcal A_j\).  Suppose that a posterior distribution is given in
factorized form
\begin{equation}
  P(\boldsymbol{u})
  =
  \frac{1}{Z}
  \prod_{a\in F} f_a(\boldsymbol{u}_{\partial a}),
  \label{eq:generic_factorization}
\end{equation}
where \(F\) is an index set of factors, \(\partial a\subseteq V\) is the set of
variables on which \(f_a\) depends, and
\(\boldsymbol{u}_{\partial a}=(u_j)_{j\in\partial a}\).
This factorization defines a bipartite factor graph with variable nodes
\(V\), factor nodes \(F\), and an edge \((j,a)\) whenever \(j\in\partial a\).
For a variable node \(j\), write
\(\partial j=\{a\in F:j\in\partial a\}\).

The sum-product algorithm is a local message-passing algorithm on this graph.
Messages are nonnegative functions on the alphabet of one variable, and are
identified up to positive normalization when the message is nonzero.  In the
strictly positive case this is the usual positive-message convention; zero
entries can be handled by the same probability-domain equations, or by limits
when a log-domain expression is used.  With iteration index \(\ell\), the
variable-node processing is
\begin{equation}
  m_{j\to a}^{(\ell)}(u_j)
  \propto
  \prod_{b\in\partial j\setminus\{a\}}
  \widehat m_{b\to j}^{(\ell-1)}(u_j),
  \qquad a\in\partial j.
  \label{eq:generic_variable_update}
\end{equation}
The factor-node processing is
\begin{equation}
  \widehat m_{a\to j}^{(\ell)}(u_j)
  \propto
  \sum_{\boldsymbol{u}_{\partial a\setminus\{j\}}}
  f_a(\boldsymbol{u}_{\partial a})
  \prod_{k\in\partial a\setminus\{j\}}
  m_{k\to a}^{(\ell)}(u_k),
  \qquad j\in\partial a.
  \label{eq:generic_factor_update}
\end{equation}
Finally, the variable marginal is represented by the belief
\begin{equation}
  b_j^{(\ell)}(u_j)
  \propto
  \prod_{a\in\partial j}
  \widehat m_{a\to j}^{(\ell)}(u_j),
  \label{eq:generic_final_belief}
\end{equation}
and, if needed, a factor belief is
\[
  b_a^{(\ell)}(\boldsymbol{u}_{\partial a})
  \propto
  f_a(\boldsymbol{u}_{\partial a})
  \prod_{j\in\partial a}
  m_{j\to a}^{(\ell)}(u_j).
\]
On a tree factor graph, after messages have been passed from the leaves inward,
these beliefs give the exact marginals of \eqref{eq:generic_factorization}.
On a graph with cycles, the same local equations define the loopy BP
iteration.  The joint BP and four-state BP decoders below are obtained by
applying \eqref{eq:generic_variable_update}--\eqref{eq:generic_final_belief} to
two different factorizations of the same posterior.

\section{Posterior Factorization}

We work directly with a binary CSS code.  Let \(H^X\) and \(H^Z\) be the two
binary check matrices.  The superscript denotes the check type: rows of
\(H^X\) are \(X\)-type checks and rows of \(H^Z\) are \(Z\)-type checks.
Let \(H^X,H^Z\in\mathbb F_2^{m\times n}\) be binary CSS check matrices, and assume
the CSS orthogonality condition
\[
  H^X(H^Z)^{\mathsf T}=0.
\]
Let \(V\) be the set of \(\mathbb F_4\) variable nodes, in bijection with the
column-index set \(\{1,\ldots,n\}\) of \(H^X\) and \(H^Z\).  The row-index set
of each matrix is \(\{1,\ldots,m\}\).  The coupled binary factor graph has two
binary variable-node sets \(X\) and \(Z\), both indexed by \(\{1,\ldots,n\}\).
Thus a column number, the corresponding column, and the corresponding variable
node are denoted by the same symbol.  In formulas over a graph neighborhood, we
also identify a variable node with the value carried by that node.  For example,
\(x\in X\) denotes both a column of \(H^Z\) and the binary value on the
corresponding \(X\)-variable node; similarly \(z\in Z\) denotes both a column
of \(H^X\) and the binary value on the corresponding \(Z\)-variable node.  After
the \(\mathbb F_4\) labeling below is chosen, the four-state value carried by
node \(j\in\{1,\ldots,n\}\) is denoted by \(\alpha_j\).

For a row index \(i\in\{1,\ldots,m\}\), define the column neighborhoods
\[
  \partial_i^X=\{j\in\{1,\ldots,n\}:H^X_{ij}=1\},
  \qquad
  \partial_i^Z=\{j\in\{1,\ldots,n\}:H^Z_{ij}=1\}.
\]
For a column index \(j\in\{1,\ldots,n\}\), define the check neighborhoods
\[
  \partial^X_j=\{i\in\{1,\ldots,m\}:H^X_{ij}=1\},
  \qquad
  \partial^Z_j=\{i\in\{1,\ldots,m\}:H^Z_{ij}=1\}.
\]
The superscript indicates the check type: an \(X\)-type check constrains the
\(z\)-variables, and a \(Z\)-type check constrains the \(x\)-variables.
When the context is the binary factor graph, the same symbols
\(\partial_i^X\) and \(\partial_i^Z\) denote adjacent binary variable-node sets.
In the four-state factor graph, they denote the corresponding adjacent column
indices of four-state variable nodes in \(V\).
Represent a Pauli error by two binary assignments
\(\boldsymbol{x},\boldsymbol{z}\in\mathbb F_2^n\).  With the above
identification, the syndrome equations are
\[
  H^X\boldsymbol{z}=\boldsymbol{s}^Z,\qquad
  H^Z\boldsymbol{x}=\boldsymbol{s}^X,
\]
or equivalently
\[
  s_i^Z=\sum_{j\in\partial_i^X}z_j\quad(1\le i\le m),
  \qquad
s_i^X=\sum_{j\in\partial_i^Z}x_j\quad(1\le i\le m).
\]
Here the superscript on the syndrome denotes the detected error component, not
the measured check type.
Let \(Q_j(x,z)\) be the local prior distribution of
\((x_j,z_j)\in\mathbb F_2^2\) at \(j\in\{1,\ldots,n\}\).  For the depolarizing channel
with physical error probability \(p\),
\[
  Q_j(0,0)=1-p,\qquad
  Q_j(1,0)=Q_j(0,1)=Q_j(1,1)=p/3,
\]
but the definitions below allow arbitrary local priors.

We now define three decoders by specifying the posterior factorization to which
ordinary sum-product is applied.  The first definition is the natural binary
factor graph for CSS syndrome decoding.  It keeps the two binary Tanner graphs
visible, but it does not separate the two Pauli components because the local
factor \(Q_j(x_j,z_j)\) remains a joint prior.

\begin{definition}[Joint BP decoder]
\label{def:gf2_decoder}
Joint BP denotes the coupled binary belief-propagation decoder obtained by
applying sum-product to the CSS factor graph with the joint local prior
\(Q_j(x_j,z_j)\).
The posterior in the coupled binary representation is
\begin{equation}
  P_2(\boldsymbol{x}, \boldsymbol{z}\mid \boldsymbol{s}^X, \boldsymbol{s}^Z)
  =
  \frac{1}{Z_2}
  \prod_{j=1}^n Q_j\left(x_j, z_j\right)
  \prod_{i=1}^m
  \mathbf{1}\!\left\{\sum_{j'\in\partial_i^X}z_{j'}=s_i^Z\right\}
  \prod_{i=1}^m
  \mathbf{1}\!\left\{\sum_{j'\in\partial_i^Z}x_{j'}=s_i^X\right\}.
  \label{eq:posterior_factorization}
\end{equation}
All sums in the indicator functions are over \(\mathbb F_2\).

The associated factor graph has binary variable-node sets
\(X\) and \(Z\), each indexed by \(\{1,\ldots,n\}\).  A \(Z\)-type check \(i\) is connected
only to the nodes \(x_j\) with \(j\in\partial_i^Z\), and an \(X\)-type check
\(i\) is connected only to the nodes \(z_j\) with \(j\in\partial_i^X\).  The factor
\(Q_j(x_j,z_j)\) couples the two binary variables indexed by the same
\(j\in\{1,\ldots,n\}\).  Ordinary sum-product on this graph is joint BP.
\end{definition}

This definition is the binary factor-graph representation that the note
advocates for CSS codes.  It is binary at the check nodes and in the messages
sent to the check nodes, so the usual Tanner-graph viewpoint is retained.  At
the same time, the prior at each qubit is still the joint table \(Q_j\), so the
factorization contains the local \(X/Z\) channel correlation whenever the
channel has one.

For comparison, we also name the separate BP decoder.  This decoder is a
useful baseline because it is often what is meant when binary BP is contrasted
with four-state BP in the quantum-code decoding literature.  Its defining
feature is not the use of binary messages, but the replacement of the joint
local prior by two marginal priors.

\begin{definition}[Separate BP decoder]
\label{def:gf2_separated_decoder}
Separate BP denotes the decoder obtained by separating the two binary posterior
factors.
Define the marginal local priors
\[
  Q_j^X(x)=\sum_{z\in\mathbb F_2}Q_j(x,z),\qquad
  Q_j^Z(z)=\sum_{x\in\mathbb F_2}Q_j(x,z).
\]
The separate BP decoder replaces the joint local factor
\(Q_j(x_j,z_j)\) by the two marginal factors \(Q_j^X(x_j)\) and
\(Q_j^Z(z_j)\).  Equivalently, it applies ordinary sum-product separately to
the two binary posterior factors
\begin{align}
  P_X(\boldsymbol{x}\mid\boldsymbol{s}^X)
  &=
  \frac{1}{Z_X}
  \prod_{j=1}^n Q_j^X(x_j)
  \prod_{i=1}^m
  \mathbf{1}\!\left\{\sum_{j'\in\partial_i^Z}x_{j'}=s_i^X\right\},
  \label{eq:separated_x_posterior}
  \\
  P_Z(\boldsymbol{z}\mid\boldsymbol{s}^Z)
  &=
  \frac{1}{Z_Z}
  \prod_{j=1}^n Q_j^Z(z_j)
  \prod_{i=1}^m
  \mathbf{1}\!\left\{\sum_{j'\in\partial_i^X}z_{j'}=s_i^Z\right\}.
  \label{eq:separated_z_posterior}
\end{align}
Thus the \(X\)-component and \(Z\)-component decoders are two disconnected
binary Tanner-graph decoders.  This decoder is binary in both alphabet and
factor graph, but it is not the same algorithm as joint BP: the
local \(X/Z\) channel correlation in \(Q_j(x_j,z_j)\) has been discarded by
marginalization.
\end{definition}

This definition isolates the source of the approximation made by separate
binary decoding.  The graph has two standard binary Tanner decoders, but there
is no factor through which information about the \(x_j\) component can affect
the \(z_j\) component, or conversely.  Therefore separate BP is not
the right object to compare with four-state BP when the question is whether a
binary factor graph can use the joint prior \(Q_j(x_j,z_j)\).

The third definition rewrites the same CSS posterior with one four-state label
per qubit.  This is the form closest to the usual Pauli-label or
\(\mathbb F_4\)-labeled description.  The purpose of the definition is to make
clear that the four labels are used here as a relabeling of the pair
\((x_j,z_j)\), not as a change of the CSS code itself.

\begin{definition}[Four-state BP decoder]
\label{def:gf4_decoder}
Write
\[
  \mathbb F_4=\{0,1,\omega,\omega^2\},\qquad \omega^2=\omega+1.
\]
Use the standard labeling
\[
  \phi:\mathbb F_2^2\to\mathbb F_4,\qquad
  \phi(x,z)=x+\omega z.
\]
Thus \(0,1,\omega,\omega^2\) correspond to
\((0,0),(1,0),(0,1),(1,1)\), respectively.  For
\(\alpha\in\mathbb F_4\), write
\(\alpha=x(\alpha)+\omega z(\alpha)\), equivalently
\(\phi^{-1}(\alpha)=(x(\alpha),z(\alpha))\).  For each \(j\in\{1,\ldots,n\}\), define
\[
  Q_j^\phi(\alpha)=Q_j(x(\alpha),z(\alpha)).
\]
The \(\mathbb F_4\)-labeled posterior is
\begin{align}
  P_4(\boldsymbol{\alpha}
  \mid \boldsymbol{s}^X,\boldsymbol{s}^Z)
  &=
  \frac{1}{Z_4}
  \prod_{j=1}^n Q_j^\phi(\alpha_j)
  \prod_{i=1}^m
  \mathbf 1\!\left\{\sum_{j\in\partial_i^X}
  z(\alpha_j)=s_i^Z\right\}
  \notag\\
  &\qquad{}\times
  \prod_{i=1}^m
  \mathbf 1\!\left\{\sum_{j\in\partial_i^Z}
  x(\alpha_j)=s_i^X\right\}.
  \label{eq:gf4_posterior_factorization}
\end{align}
Sum-product on this factor graph is called the four-state BP decoder.
The choice \(\phi(x,z)=x+\omega z\) only fixes which \(\mathbb F_4\) labels
denote the \(X\)- and \(Z\)-components.
\end{definition}

This definition puts the four-state decoder on the same footing as the coupled
binary decoder: both are ordinary sum-product algorithms applied to explicitly
written posterior factorizations.  The only structural difference is whether
the local pair \((x_j,z_j)\) is represented by two binary variables connected
through \(Q_j\), or by the single label \(\alpha_j=x_j+\omega z_j\).

The three decoders considered in this note differ by their posterior
factorizations.  Separate BP uses the two marginal priors
\(Q_j^X(x_j)\) and \(Q_j^Z(z_j)\) and therefore has two disconnected binary
factor graphs.  Joint BP uses the joint local prior
\(Q_j(x_j,z_j)\), so the two binary Tanner graphs are connected through a
local factor at each qubit.  Four-state BP uses one four-state variable
\(\alpha_j\) per qubit and the equivalent local prior
\(Q_j^\phi(\alpha_j)\).  The equivalence theorem below compares
joint BP with four-state BP.  It does not compare either of them with
separate BP.

\section{Example: A Randomized Length-24 \texorpdfstring{\((2,6)\)}{(2,6)}-Regular CSS Pair}

For a \((2,6)\)-regular binary parity-check matrix, the edge-count constraint is
\(2n=6m\).  Keeping the same code length \(n=24\) gives \(m=8\).
The following CSS matrix pair is a length-24 example obtained within the
randomized construction framework of \cite{okadaKasai25RandomConstruction},
specialized to the \((2,6)\)-regular degree profile.  It satisfies
\[
  H^X(H^Z)^{\mathsf T}=0,\qquad
  \mathrm{wt}_{\mathrm{col}}(H^X)=\mathrm{wt}_{\mathrm{col}}(H^Z)=2,\qquad
  \mathrm{wt}_{\mathrm{row}}(H^X)=\mathrm{wt}_{\mathrm{row}}(H^Z)=6.
\]
The row supports define the two check matrices:
\[
\begin{array}{c@{\quad}l@{\qquad}c@{\quad}l}
i & \partial_i^X & i & \partial_i^Z\\ \hline
1&\{10,14,17,19,21,22\}&1&\{5,15,17,18,20,22\}\\
2&\{2,3,4,7,8,24\}&2&\{2,8,12,13,14,19\}\\
3&\{3,4,6,9,11,16\}&3&\{1,3,4,10,14,18\}\\
4&\{1,5,11,12,13,18\}&4&\{4,9,10,21,23,24\}\\
5&\{5,9,10,12,14,17\}&5&\{3,8,9,15,17,19\}\\
6&\{2,6,8,15,16,20\}&6&\{1,7,11,16,20,24\}\\
7&\{1,7,18,20,23,24\}&7&\{2,6,7,11,13,23\}\\
8&\{13,15,19,21,22,23\}&8&\{5,6,12,16,21,22\}.
\end{array}
\]
The rows of \(H^X\) constrain the \(z\)-assignment, while the rows of \(H^Z\)
constrain the \(x\)-assignment.
Equivalently, the matrices themselves are
displayed below.  The highlighted entries \(H^X_{1,10}=1\) and \(H^Z_{1,5}=1\)
are the two representative edges highlighted again in the factor graphs.
{\scriptsize
\[
\setlength{\arraycolsep}{1.4pt}
H^X=
\left[
\begin{array}{*{24}{c}}
\matzero&\matzero&\matzero&\matzero&\matzero&\matzero&\matzero&\matzero&\matzero&\hxedgeone&\matzero&\matzero&\matzero&1&\matzero&\matzero&1&\matzero&1&\matzero&1&1&\matzero&\matzero\\
\matzero&1&1&1&\matzero&\matzero&1&1&\matzero&\matzero&\matzero&\matzero&\matzero&\matzero&\matzero&\matzero&\matzero&\matzero&\matzero&\matzero&\matzero&\matzero&\matzero&1\\
\matzero&\matzero&1&1&\matzero&1&\matzero&\matzero&1&\matzero&1&\matzero&\matzero&\matzero&\matzero&1&\matzero&\matzero&\matzero&\matzero&\matzero&\matzero&\matzero&\matzero\\
1&\matzero&\matzero&\matzero&1&\matzero&\matzero&\matzero&\matzero&\matzero&1&1&1&\matzero&\matzero&\matzero&\matzero&1&\matzero&\matzero&\matzero&\matzero&\matzero&\matzero\\
\matzero&\matzero&\matzero&\matzero&1&\matzero&\matzero&\matzero&1&1&\matzero&1&\matzero&1&\matzero&\matzero&1&\matzero&\matzero&\matzero&\matzero&\matzero&\matzero&\matzero\\
\matzero&1&\matzero&\matzero&\matzero&1&\matzero&1&\matzero&\matzero&\matzero&\matzero&\matzero&\matzero&1&1&\matzero&\matzero&\matzero&1&\matzero&\matzero&\matzero&\matzero\\
1&\matzero&\matzero&\matzero&\matzero&\matzero&1&\matzero&\matzero&\matzero&\matzero&\matzero&\matzero&\matzero&\matzero&\matzero&\matzero&1&\matzero&1&\matzero&\matzero&1&1\\
\matzero&\matzero&\matzero&\matzero&\matzero&\matzero&\matzero&\matzero&\matzero&\matzero&\matzero&\matzero&1&\matzero&1&\matzero&\matzero&\matzero&1&\matzero&1&1&1&\matzero
\end{array}
\right],
\]
}
{\scriptsize
\[
\setlength{\arraycolsep}{1.4pt}
H^Z=
\left[
\begin{array}{*{24}{c}}
\matzero&\matzero&\matzero&\matzero&\hzedgeone&\matzero&\matzero&\matzero&\matzero&\matzero&\matzero&\matzero&\matzero&\matzero&1&\matzero&1&1&\matzero&1&\matzero&1&\matzero&\matzero\\
\matzero&1&\matzero&\matzero&\matzero&\matzero&\matzero&1&\matzero&\matzero&\matzero&1&1&1&\matzero&\matzero&\matzero&\matzero&1&\matzero&\matzero&\matzero&\matzero&\matzero\\
1&\matzero&1&1&\matzero&\matzero&\matzero&\matzero&\matzero&1&\matzero&\matzero&\matzero&1&\matzero&\matzero&\matzero&1&\matzero&\matzero&\matzero&\matzero&\matzero&\matzero\\
\matzero&\matzero&\matzero&1&\matzero&\matzero&\matzero&\matzero&1&1&\matzero&\matzero&\matzero&\matzero&\matzero&\matzero&\matzero&\matzero&\matzero&\matzero&1&\matzero&1&1\\
\matzero&\matzero&1&\matzero&\matzero&\matzero&\matzero&1&1&\matzero&\matzero&\matzero&\matzero&\matzero&1&\matzero&1&\matzero&1&\matzero&\matzero&\matzero&\matzero&\matzero\\
1&\matzero&\matzero&\matzero&\matzero&\matzero&1&\matzero&\matzero&\matzero&1&\matzero&\matzero&\matzero&\matzero&1&\matzero&\matzero&\matzero&1&\matzero&\matzero&\matzero&1\\
\matzero&1&\matzero&\matzero&\matzero&1&1&\matzero&\matzero&\matzero&1&\matzero&1&\matzero&\matzero&\matzero&\matzero&\matzero&\matzero&\matzero&\matzero&\matzero&1&\matzero\\
\matzero&\matzero&\matzero&\matzero&1&1&\matzero&\matzero&\matzero&\matzero&\matzero&1&\matzero&\matzero&\matzero&1&\matzero&\matzero&\matzero&\matzero&1&1&\matzero&\matzero
\end{array}
\right].
\]
}
This displayed support table is also a compact certificate for the stated
CSS orthogonality condition: every row has size \(6\), every column index \(j\in\{1,\ldots,24\}\)
appears exactly twice in the \(H^X\) table and exactly twice in the
\(H^Z\) table, and every intersection
\(\partial_i^X\cap\partial_k^Z\) has even cardinality for all \(i,k\).
The last condition is equivalent to \((H^X(H^Z)^{\mathsf T})_{ik}=0\) over
\(\mathbb F_2\).  In this example the intersection sizes are \(0\), \(2\), or \(4\),
with multiplicities \(17,46\), and \(1\), respectively.

Figure~\ref{fig:example_gf2_factor_graph} shows the coupled \(\mathbb F_2\) factor graph.
The \(Z\)-type checks are attached only to variables \(x_j\), the \(X\)-type checks
are attached only to variables \(z_j\), and the factors \(Q_j\) couple
\((x_j,z_j)\) locally at each \(j\in\{1,\ldots,n\}\).
In the figures, the label \(s_i^X\) denotes the \(Z\)-type syndrome-conditioned
check factor \(\mathbf{1}\{\sum_{j\in\partial_i^Z}x_j=s_i^X\}\), while
\(s_i^Z\) denotes the \(X\)-type syndrome-conditioned check factor
\(\mathbf{1}\{\sum_{j\in\partial_i^X}z_j=s_i^Z\}\).

\begin{figure}[t]
\centering
\resizebox{\linewidth}{!}{%
\begin{tikzpicture}[
  x=0.66cm,
  y=1cm,
  var/.style={circle,draw,minimum size=1.54mm,inner sep=0pt},
  xvar/.style={var,draw=blue!55!black,fill=blue!18},
  zvar/.style={var,draw=green!45!black,fill=green!18},
  f4var/.style={var,draw=violet!65!black,fill=violet!18},
  local/.style={rectangle,draw=orange!70!black,fill=orange!22,minimum size=1.39mm,inner sep=0pt},
  factor/.style={rectangle,draw,minimum size=1.39mm,inner sep=0pt},
  zcheck/.style={factor,draw=red!65!black,fill=red!16},
  xcheck/.style={factor,draw=cyan!65!black,fill=cyan!18},
  lab/.style={font=\tiny,inner sep=0.2pt},
  edge/.style={draw=black!38,line width=0.25pt},
  hzed/.style={draw=red!75!black,line width=1.05pt},
  hxed/.style={draw=cyan!75!black,line width=1.05pt},
  every node/.style={font=\scriptsize}
]
  \foreach \i/\x in {1/0,2/1,3/2,4/3,5/4,6/5,7/6,8/7,9/8,10/9,11/10,12/11,
    13/12,14/13,15/14,16/15,17/16,18/17,19/18,20/19,21/20,22/21,23/22,24/23} {
    \node[xvar] (X\i) at (\x,2.0) {};
    \node[local] (Q\i) at (\x,1.0) {};
    \node[zvar] (Z\i) at (\x,0.0) {};
    \draw (X\i) -- (Q\i);
    \draw (Q\i) -- (Z\i);
  }
  \node[lab,left=1.5pt of X1] {$x_j$};
  \node[lab,left=1.5pt of Q1] {$Q_j$};
  \node[lab,left=1.5pt of Z1] {$z_j$};
  \foreach \j/\x in {1/0.5,2/3.7,3/6.9,4/10.1,5/13.3,6/16.5,7/19.7,8/22.9} {
    \node[zcheck] (HZ\j) at (\x,4.295) {};
    \node[xcheck] (HX\j) at (\x,-2.295) {};
  }
  \node[lab,left=1.5pt of HZ1] {$s_i^X$};
  \node[lab,left=1.5pt of HX1] {$s_i^Z$};
  \node[lab,red!70!black] at (23.9,3.2) {$H^Z$};
  \node[lab,cyan!70!black] at (23.9,-1.2) {$H^X$};
  \foreach \i in {5,15,17,18,20,22} {\draw[edge] (HZ1.south) -- (X\i.north);}
  \foreach \i in {2,8,12,13,14,19} {\draw[edge] (HZ2.south) -- (X\i.north);}
  \foreach \i in {1,3,4,10,14,18} {\draw[edge] (HZ3.south) -- (X\i.north);}
  \foreach \i in {4,9,10,21,23,24} {\draw[edge] (HZ4.south) -- (X\i.north);}
  \foreach \i in {3,8,9,15,17,19} {\draw[edge] (HZ5.south) -- (X\i.north);}
  \foreach \i in {1,7,11,16,20,24} {\draw[edge] (HZ6.south) -- (X\i.north);}
  \foreach \i in {2,6,7,11,13,23} {\draw[edge] (HZ7.south) -- (X\i.north);}
  \foreach \i in {5,6,12,16,21,22} {\draw[edge] (HZ8.south) -- (X\i.north);}
  \foreach \i in {10,14,17,19,21,22} {\draw[edge] (HX1.north) -- (Z\i.south);}
  \foreach \i in {2,3,4,7,8,24} {\draw[edge] (HX2.north) -- (Z\i.south);}
  \foreach \i in {3,4,6,9,11,16} {\draw[edge] (HX3.north) -- (Z\i.south);}
  \foreach \i in {1,5,11,12,13,18} {\draw[edge] (HX4.north) -- (Z\i.south);}
  \foreach \i in {5,9,10,12,14,17} {\draw[edge] (HX5.north) -- (Z\i.south);}
  \foreach \i in {2,6,8,15,16,20} {\draw[edge] (HX6.north) -- (Z\i.south);}
  \foreach \i in {1,7,18,20,23,24} {\draw[edge] (HX7.north) -- (Z\i.south);}
  \foreach \i in {13,15,19,21,22,23} {\draw[edge] (HX8.north) -- (Z\i.south);}
  \draw[hzed] (HZ1.south) -- (X5.north);
  \draw[hxed] (HX1.north) -- (Z10.south);
\end{tikzpicture}
}
\caption{Coupled \(\mathbb F_2\) factor graph for the length-24 \((2,6)\)-regular
CSS pair.  The labels \(s_i^X\) and \(s_i^Z\) denote the syndrome-conditioned
check factors, and the middle factors \(Q_j(x_j,z_j)\) carry the local X/Z
correlation.  The highlighted edges correspond to \(H^Z_{1,5}=1\) and
\(H^X_{1,10}=1\).}
\label{fig:example_gf2_factor_graph}
\end{figure}
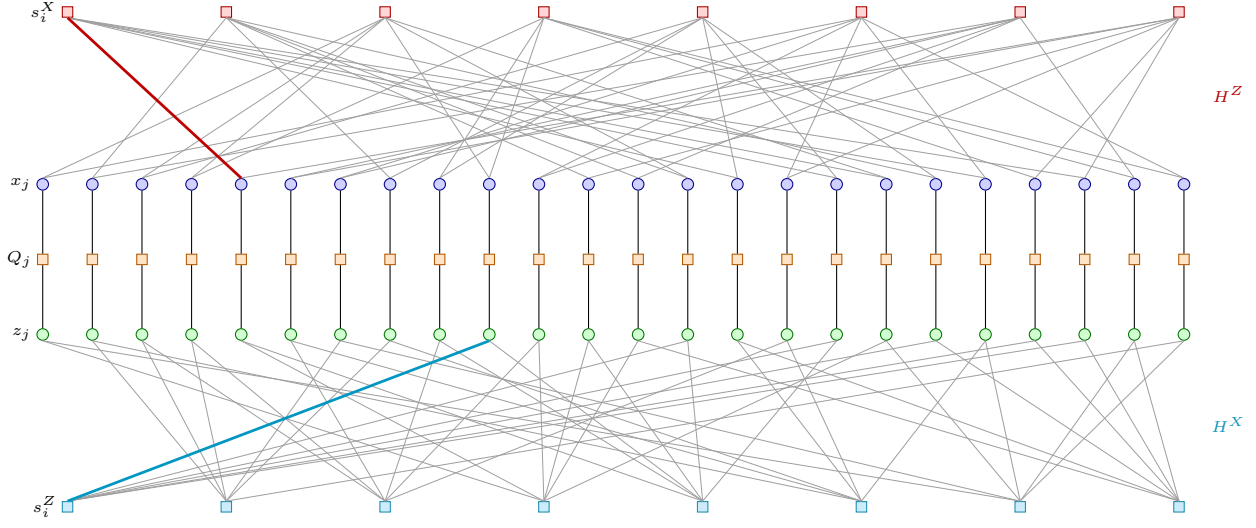

Figure~\ref{fig:separated_gf2_factor_graph} shows the separate BP
decoder that ignores the local \(X/Z\) correlation.  In this decoder, the joint
factor \(Q_j(x_j,z_j)\) is replaced by its two marginals
\[
  Q_j^X(x_j)=\sum_{z_j\in\mathbb F_2}Q_j(x_j,z_j),\qquad
  Q_j^Z(z_j)=\sum_{x_j\in\mathbb F_2}Q_j(x_j,z_j).
\]
The \(x\)- and \(z\)-decoders are then two disconnected Tanner graphs.  This is
the usual separate BP baseline, not joint BP
discussed in this note.

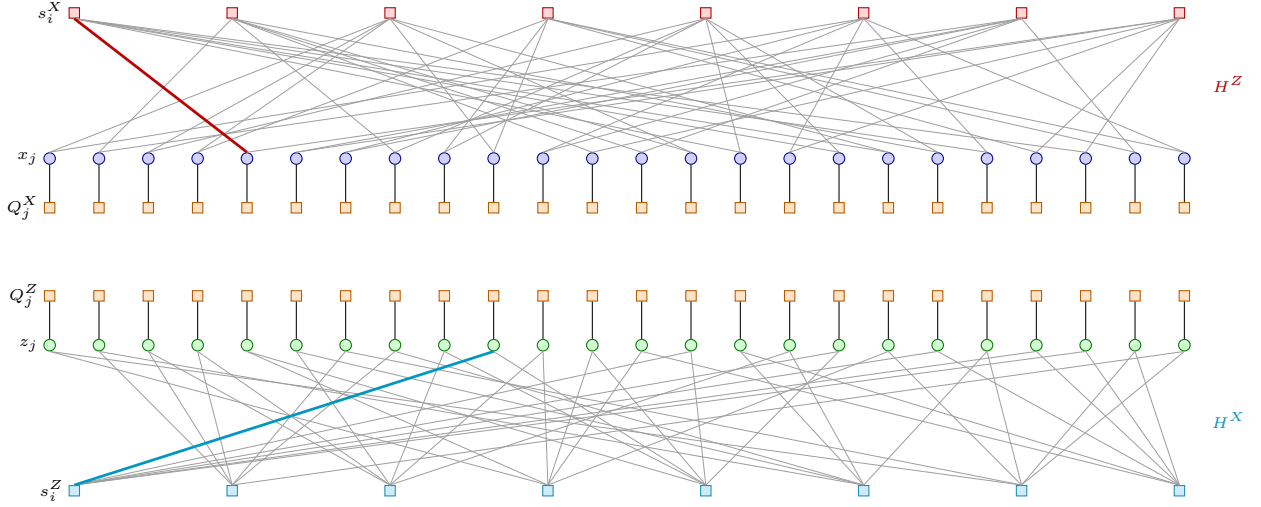
\begin{figure}[t]
\centering
\resizebox{\linewidth}{!}{%
\begin{tikzpicture}[
  x=0.66cm,
  y=1cm,
  var/.style={circle,draw,minimum size=1.54mm,inner sep=0pt},
  xvar/.style={var,draw=blue!55!black,fill=blue!18},
  zvar/.style={var,draw=green!45!black,fill=green!18},
  local/.style={rectangle,draw=orange!70!black,fill=orange!22,minimum size=1.39mm,inner sep=0pt},
  factor/.style={rectangle,draw,minimum size=1.39mm,inner sep=0pt},
  zcheck/.style={factor,draw=red!65!black,fill=red!16},
  xcheck/.style={factor,draw=cyan!65!black,fill=cyan!18},
  lab/.style={font=\tiny,inner sep=0.2pt},
  edge/.style={draw=black!38,line width=0.25pt},
  hzed/.style={draw=red!75!black,line width=1.05pt},
  hxed/.style={draw=cyan!75!black,line width=1.05pt},
  every node/.style={font=\scriptsize}
]
  \foreach \i/\x in {1/0,2/1,3/2,4/3,5/4,6/5,7/6,8/7,9/8,10/9,11/10,12/11,
    13/12,14/13,15/14,16/15,17/16,18/17,19/18,20/19,21/20,22/21,23/22,24/23} {
    \node[xvar] (SX\i) at (\x,1.0) {};
    \node[local] (QX\i) at (\x,0.34) {};
    \draw (SX\i) -- (QX\i);
    \node[local] (QZ\i) at (\x,-0.84) {};
    \node[zvar] (SZ\i) at (\x,-1.5) {};
    \draw (QZ\i) -- (SZ\i);
  }
  \node[lab,left=1.5pt of SX1] {$x_j$};
  \node[lab,left=1.5pt of QX1] {$Q_j^X$};
  \node[lab,left=1.5pt of QZ1] {$Q_j^Z$};
  \node[lab,left=1.5pt of SZ1] {$z_j$};
  \foreach \j/\x in {1/0.5,2/3.7,3/6.9,4/10.1,5/13.3,6/16.5,7/19.7,8/22.9} {
    \node[zcheck] (SHZ\j) at (\x,2.95) {};
    \node[xcheck] (SHX\j) at (\x,-3.45) {};
  }
  \node[lab,left=1.5pt of SHZ1] {$s_i^X$};
  \node[lab,left=1.5pt of SHX1] {$s_i^Z$};
  \node[lab,red!70!black] at (23.9,2.0) {$H^Z$};
  \node[lab,cyan!70!black] at (23.9,-2.5) {$H^X$};
  \foreach \i in {5,15,17,18,20,22} {\draw[edge] (SHZ1.south) -- (SX\i.north);}
  \foreach \i in {2,8,12,13,14,19} {\draw[edge] (SHZ2.south) -- (SX\i.north);}
  \foreach \i in {1,3,4,10,14,18} {\draw[edge] (SHZ3.south) -- (SX\i.north);}
  \foreach \i in {4,9,10,21,23,24} {\draw[edge] (SHZ4.south) -- (SX\i.north);}
  \foreach \i in {3,8,9,15,17,19} {\draw[edge] (SHZ5.south) -- (SX\i.north);}
  \foreach \i in {1,7,11,16,20,24} {\draw[edge] (SHZ6.south) -- (SX\i.north);}
  \foreach \i in {2,6,7,11,13,23} {\draw[edge] (SHZ7.south) -- (SX\i.north);}
  \foreach \i in {5,6,12,16,21,22} {\draw[edge] (SHZ8.south) -- (SX\i.north);}
  \foreach \i in {10,14,17,19,21,22} {\draw[edge] (SHX1.north) -- (SZ\i.south);}
  \foreach \i in {2,3,4,7,8,24} {\draw[edge] (SHX2.north) -- (SZ\i.south);}
  \foreach \i in {3,4,6,9,11,16} {\draw[edge] (SHX3.north) -- (SZ\i.south);}
  \foreach \i in {1,5,11,12,13,18} {\draw[edge] (SHX4.north) -- (SZ\i.south);}
  \foreach \i in {5,9,10,12,14,17} {\draw[edge] (SHX5.north) -- (SZ\i.south);}
  \foreach \i in {2,6,8,15,16,20} {\draw[edge] (SHX6.north) -- (SZ\i.south);}
  \foreach \i in {1,7,18,20,23,24} {\draw[edge] (SHX7.north) -- (SZ\i.south);}
  \foreach \i in {13,15,19,21,22,23} {\draw[edge] (SHX8.north) -- (SZ\i.south);}
  \draw[hzed] (SHZ1.south) -- (SX5.north);
  \draw[hxed] (SHX1.north) -- (SZ10.south);
\end{tikzpicture}
}
\caption{Separated \(\mathbb F_2\) Tanner graphs obtained by discarding the local
X/Z correlation.  The upper graph decodes \(x\) from \(H^Z\boldsymbol{x}=\boldsymbol{s}^X\)
using \(Q_j^X\), and the lower graph decodes \(z\) from
\(H^X\boldsymbol{z}=\boldsymbol{s}^Z\) using \(Q_j^Z\).  There is no factor
connecting \(x_j\) and \(z_j\).}
\label{fig:separated_gf2_factor_graph}
\end{figure}

Figure~\ref{fig:example_gf4_factor_graph} shows the \(\mathbb F_4\)-labeled
four-state factor graph for the same posterior.  The variable
\(\alpha_j=x_j+\omega z_j\) is the value carried by the node \(j\in\{1,\ldots,n\}\).
The same parity constraints now read only the appropriate binary
coordinate of each \(\alpha_j\).  The check-node labels \(s_i^X\) and \(s_i^Z\)
have the same meaning as in Figure~\ref{fig:example_gf2_factor_graph}.

\begin{figure}[t]
\centering
\resizebox{\linewidth}{!}{%
\begin{tikzpicture}[
  x=0.66cm,
  y=1cm,
  var/.style={circle,draw,minimum size=1.54mm,inner sep=0pt},
  f4var/.style={var,draw=violet!65!black,fill=violet!18},
  local/.style={rectangle,draw=orange!70!black,fill=orange!22,minimum size=1.39mm,inner sep=0pt},
  factor/.style={rectangle,draw,minimum size=1.39mm,inner sep=0pt},
  zcheck/.style={factor,draw=red!65!black,fill=red!16},
  xcheck/.style={factor,draw=cyan!65!black,fill=cyan!18},
  lab/.style={font=\tiny,inner sep=0.2pt},
  edge/.style={draw=black!38,line width=0.25pt},
  hzed/.style={draw=red!75!black,line width=1.05pt},
  hxed/.style={draw=cyan!75!black,line width=1.05pt},
  every node/.style={font=\scriptsize}
]
  \foreach \i/\x in {1/0,2/1,3/2,4/3,5/4,6/5,7/6,8/7,9/8,10/9,11/10,12/11,
    13/12,14/13,15/14,16/15,17/16,18/17,19/18,20/19,21/20,22/21,23/22,24/23} {
    \node[f4var] (A\i) at (\x,0) {};
    \node[local] (Q\i) at ({\x-0.76},0.56) {};
    \draw (Q\i) -- (A\i);
  }
  \node[lab,left=1.5pt of A1] {$\alpha_j$};
  \node[lab,left=1.5pt of Q1] {$Q_j^\phi$};
  \foreach \j/\x in {1/0.5,2/3.7,3/6.9,4/10.1,5/13.3,6/16.5,7/19.7,8/22.9} {
    \node[zcheck] (HZ\j) at (\x,2.6) {};
    \node[xcheck] (HX\j) at (\x,-2.0) {};
  }
  \node[lab,left=1.5pt of HZ1] {$s_i^X$};
  \node[lab,left=1.5pt of HX1] {$s_i^Z$};
  \node[lab,red!70!black] at (23.9,1.43) {$H^Z$};
  \node[lab,cyan!70!black] at (23.9,-1.1) {$H^X$};
  \foreach \i in {5,15,17,18,20,22} {\draw[edge] (HZ1.south) -- (A\i.north);}
  \foreach \i in {2,8,12,13,14,19} {\draw[edge] (HZ2.south) -- (A\i.north);}
  \foreach \i in {1,3,4,10,14,18} {\draw[edge] (HZ3.south) -- (A\i.north);}
  \foreach \i in {4,9,10,21,23,24} {\draw[edge] (HZ4.south) -- (A\i.north);}
  \foreach \i in {3,8,9,15,17,19} {\draw[edge] (HZ5.south) -- (A\i.north);}
  \foreach \i in {1,7,11,16,20,24} {\draw[edge] (HZ6.south) -- (A\i.north);}
  \foreach \i in {2,6,7,11,13,23} {\draw[edge] (HZ7.south) -- (A\i.north);}
  \foreach \i in {5,6,12,16,21,22} {\draw[edge] (HZ8.south) -- (A\i.north);}
  \foreach \i in {10,14,17,19,21,22} {\draw[edge] (HX1.north) -- (A\i.south);}
  \foreach \i in {2,3,4,7,8,24} {\draw[edge] (HX2.north) -- (A\i.south);}
  \foreach \i in {3,4,6,9,11,16} {\draw[edge] (HX3.north) -- (A\i.south);}
  \foreach \i in {1,5,11,12,13,18} {\draw[edge] (HX4.north) -- (A\i.south);}
  \foreach \i in {5,9,10,12,14,17} {\draw[edge] (HX5.north) -- (A\i.south);}
  \foreach \i in {2,6,8,15,16,20} {\draw[edge] (HX6.north) -- (A\i.south);}
  \foreach \i in {1,7,18,20,23,24} {\draw[edge] (HX7.north) -- (A\i.south);}
  \foreach \i in {13,15,19,21,22,23} {\draw[edge] (HX8.north) -- (A\i.south);}
  \draw[hzed] (HZ1.south) -- (A5.north);
  \draw[hxed] (HX1.north) -- (A10.south);
\end{tikzpicture}
}
\caption{\(\mathbb F_4\)-labeled factor graph for the same CSS pair.  The
four-state variable \(\alpha_j=x_j+\omega z_j\) is identified with its variable
node, and the syndrome-conditioned check factors \(s_i^Z\) and \(s_i^X\) use
\(x(\alpha_j)\) or \(z(\alpha_j)\) as required.  The highlighted edges are the
same two matrix entries \(H^Z_{1,5}=1\) and \(H^X_{1,10}=1\).}
\label{fig:example_gf4_factor_graph}
\end{figure}
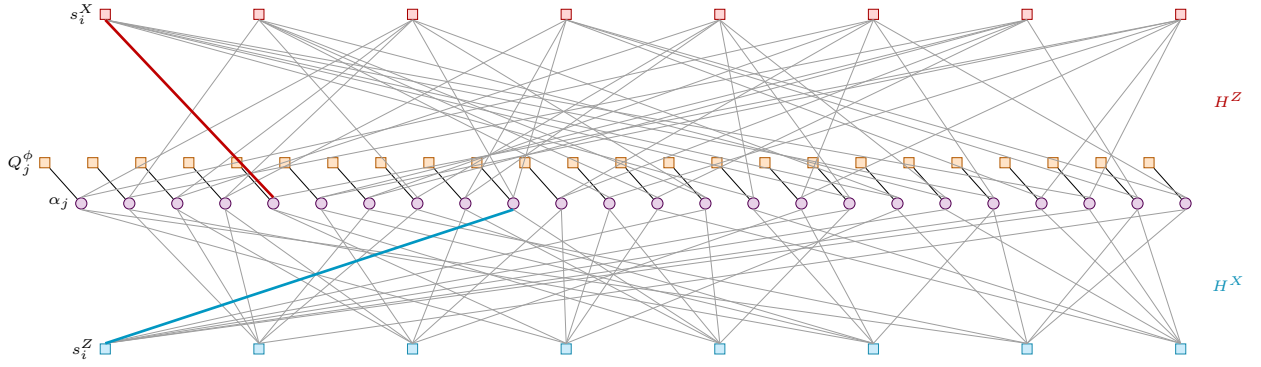

\section{Equivalence}

For \(j\in\{1,\ldots,n\}\) and \(i\in\partial^X_j\), let
\(m_{j\to i}^{X,(\ell)}(\alpha)\) and
\(\widehat{m}_{i\to j}^{X,(\ell)}(\alpha)\) be the \(\mathbb F_4\)-labeled
variable-to-check and check-to-variable messages on the edge between
\(\alpha_j\) and the \(X\)-type check \(i\).  For \(i\in\partial^Z_j\), define
\(m_{j\to i}^{Z,(\ell)}(\alpha)\) and
\(\widehat{m}_{i\to j}^{Z,(\ell)}(\alpha)\) analogously.  In the coupled binary
graph, write \(\widehat{\nu}_{i\to j}^{(\ell)}(z)\) for an \(X\)-type check
message to \(z_j\), and \(\widehat{\mu}_{i\to j}^{(\ell)}(x)\) for a
\(Z\)-type check message to \(x_j\).  Let \(\nu_{j\to i}^{(\ell)}(z)\) and
\(\mu_{j\to i}^{(\ell)}(x)\) denote the corresponding binary messages from
\(z_j\) and \(x_j\) to the checks.
The joint qubit belief associated with the local factor \(Q_j\), and the
corresponding \(\mathbb F_4\)-labeled qubit belief, are
\begin{align}
  b_{2,j}^{(\ell)}(x,z)
  &\propto
  Q_j(x,z)
  \prod_{i\in\partial^X_j}
  \widehat{\nu}_{i\to j}^{(\ell)}(z)
  \prod_{i\in\partial^Z_j}
  \widehat{\mu}_{i\to j}^{(\ell)}(x),
  \label{eq:gf2_belief}\\
  b_{4,j}^{(\ell)}(\alpha)
  &\propto
  Q_j^\phi(\alpha)
  \prod_{i\in\partial^X_j}
  \widehat{m}_{i\to j}^{X,(\ell)}(\alpha)
  \prod_{i\in\partial^Z_j}
  \widehat{m}_{i\to j}^{Z,(\ell)}(\alpha).
  \label{eq:gf4_belief}
\end{align}
The phrase ``initialized compatibly'' below means that the initial messages
already satisfy the marginal identities in
\eqref{eq:check_message_x_equiv}--\eqref{eq:var_message_z_equiv}.  The usual
uniform initialization of all check-to-variable messages is compatible; so is
any deterministic initialization obtained by assigning four-state messages and
then taking the binary marginals displayed in
\eqref{eq:var_message_x_equiv}--\eqref{eq:var_message_z_equiv}.
The proof below uses only the standard sum-product/distributive-law
manipulation of moving sums past products when a factor depends on only part of
a variable label \cite{ajiMcEliece00,kschischang01}.
All BP update equations are written with \(\propto\).  The hidden constants are
edge-local positive normalizations; they are useful for storing messages as
probability vectors, but their numerical values are not part of the algorithmic
content.  One may normalize each outgoing message arbitrarily, or work with
ratios/log-ratios where these constants cancel.
The exact equivalence below is a statement about these probability-domain
sum-product equations.  If one adds approximations or implementation choices
such as min-sum, damping, clipping, or quantum-code-specific post-processing,
the same equivalence is guaranteed only when the modification is transported
through the same relabeling and marginalization rules.
When a hard decision is needed, fix an arbitrary common tie-breaking order on
\(\mathbb F_2^2\), and use the induced order on \(\mathbb F_4\) through
\(\phi\).  Define the \(\mathbb F_4\)-labeled estimate by
\[
  \widehat{\alpha}_j^{(\ell)}
  =
  \operatorname*{arg\,max}_{\alpha\in\mathbb F_4}
  b_{4,j}^{(\ell)}(\alpha),
\]
and define the joint binary estimate by
\[
  (\widetilde{x}_j^{(\ell)},\widetilde{z}_j^{(\ell)})
  =
  \operatorname*{arg\,max}_{(x,z)\in\mathbb F_2^2}
  b_{2,j}^{(\ell)}(x,z).
\]
Separately, define the componentwise binary estimates from the products of
messages around the individual binary variables:
\[
  \widehat{z}_j^{(\ell)}
  =
  \operatorname*{arg\,max}_{z\in\mathbb F_2}
  \prod_{i\in\partial^X_j}
  \widehat{\nu}_{i\to j}^{(\ell)}(z),
  \qquad
  \widehat{x}_j^{(\ell)}
  =
  \operatorname*{arg\,max}_{x\in\mathbb F_2}
  \prod_{i\in\partial^Z_j}
  \widehat{\mu}_{i\to j}^{(\ell)}(x).
\]
These componentwise estimates are distinct from the joint estimate in general,
because the local factor \(Q_j(x,z)\) can correlate \(x_j\) and \(z_j\).  They
should also be distinguished from componentwise posterior-marginal decisions,
which would be
\[
  \bar{z}_j^{(\ell)}
  =
  \operatorname*{arg\,max}_{z\in\mathbb F_2}
  \sum_{x\in\mathbb F_2} b_{2,j}^{(\ell)}(x,z),
  \qquad
  \bar{x}_j^{(\ell)}
  =
  \operatorname*{arg\,max}_{x\in\mathbb F_2}
  \sum_{z\in\mathbb F_2} b_{2,j}^{(\ell)}(x,z).
\]

The next theorem is the central claim of the note.  It states that the coupled
binary decoder and the four-state decoder are not merely similar in output:
they are the same sum-product computation after relabeling the local state and
marginalizing the unused binary component of a four-state message.  This is the
precise sense in which joint BP can use the same local \(X/Z\)
correlation as four-state BP.

\begin{theorem}[Equivalence of joint BP and four-state BP]
\label{thm:gf2_gf4_equivalence}
Consider the decoders in Definitions~\ref{def:gf2_decoder} and
\ref{def:gf4_decoder}.  For every
\(\boldsymbol{x},\boldsymbol{z}\in\mathbb F_2^n\), define
\(\phi(\boldsymbol{x},\boldsymbol{z})\in\mathbb F_4^n\) by
\((\phi(\boldsymbol{x},\boldsymbol{z}))_j=\phi(x_j,z_j)\).  Then
\begin{equation}
  P_4\!\left(\phi(\boldsymbol{x},\boldsymbol{z})
  \mid \boldsymbol{s}^X,\boldsymbol{s}^Z\right)
  =
  P_2(\boldsymbol{x},\boldsymbol{z}
  \mid \boldsymbol{s}^X,\boldsymbol{s}^Z).
  \label{eq:posterior_relabeling}
\end{equation}
BP messages are identified up to positive normalization constants.  These
constants need not be computed explicitly.  If the two
decoders are initialized compatibly, then for every iteration \(\ell\),
\begin{align}
  \widehat{m}_{i\to j}^{X,(\ell)}(\phi(x,z))
  &\propto
  \widehat{\nu}_{i\to j}^{(\ell)}(z),
  \qquad i\in\partial^X_j,
  \label{eq:check_message_x_equiv}\\
  \widehat{m}_{i\to j}^{Z,(\ell)}(\phi(x,z))
  &\propto
  \widehat{\mu}_{i\to j}^{(\ell)}(x),
  \qquad i\in\partial^Z_j.
  \label{eq:check_message_z_equiv}
\end{align}
Moreover, the \(\mathbb F_4\)-labeled variable-to-check messages marginalize to the
coupled binary variable-to-check messages:
\begin{align}
  \nu_{j\to i}^{(\ell)}(z)
  &\propto
  \sum_{x\in\mathbb F_2}
  m_{j\to i}^{X,(\ell)}(\phi(x,z)),
  \qquad i\in\partial^X_j,
  \label{eq:var_message_x_equiv}\\
  \mu_{j\to i}^{(\ell)}(x)
  &\propto
  \sum_{z\in\mathbb F_2}
  m_{j\to i}^{Z,(\ell)}(\phi(x,z)),
  \qquad i\in\partial^Z_j.
  \label{eq:var_message_z_equiv}
\end{align}
The qubit beliefs satisfy
\begin{equation}
  b_{4,j}^{(\ell)}(\phi(\xi,\zeta))
  =
  b_{2,j}^{(\ell)}(\xi,\zeta),
  \qquad j\in\{1,\ldots,n\},\quad(\xi,\zeta)\in\mathbb F_2^2 .
  \label{eq:belief_equivalence}
\end{equation}
Thus the two decoders perform the same sum-product computation for the same
posterior distribution in different labels, and their hard decisions correspond
under \(\phi\):
\[
  \widehat{\alpha}_j^{(\ell)}
  =
  \phi\!\left(\widetilde{x}_j^{(\ell)},\widetilde{z}_j^{(\ell)}\right).
\]
\end{theorem}

The theorem has three parts, each corresponding to a different level of the
decoder.  First, the two factorizations assign the same posterior weight to
corresponding Pauli configurations.  Second, the local sum-product messages can
be matched edge by edge.  Third, the final beliefs and the resulting hard
decisions agree under the same relabeling.  Thus the theorem is stronger than
an equality of final estimates for one initialization; it identifies the whole
message-passing recursion.

\begin{proof}
We take the two posterior factorizations
\eqref{eq:posterior_factorization} and
\eqref{eq:gf4_posterior_factorization} as the starting point and compare the
sum-product algorithms applied to them.

Let \(\boldsymbol{\alpha}=\phi(\boldsymbol{x},\boldsymbol{z})\).  Then
\(x(\alpha_j)=x_j\), \(z(\alpha_j)=z_j\), and
\(Q_j^\phi(\alpha_j)=Q_j(x_j,z_j)\) for every \(j\in\{1,\ldots,n\}\).  Therefore the
unnormalized weights in \eqref{eq:posterior_factorization} and
\eqref{eq:gf4_posterior_factorization} agree term by term.  Since \(\phi\) is
a bijection, the normalizing constants are equal, and
\eqref{eq:posterior_relabeling} follows.

It remains to compare the BP updates.  The \(X\)-type check factors in the
\(\mathbb F_4\)-labeled representation depend on \(\alpha_j\) only through
\(z(\alpha_j)\), and the \(Z\)-type check factors depend on \(\alpha_j\) only
through \(x(\alpha_j)\).  Thus
the corresponding check-to-variable messages can be represented as functions of
only the relevant binary component.

Assume that, at some iteration, \eqref{eq:check_message_x_equiv} and
\eqref{eq:check_message_z_equiv} hold.  The \(\mathbb F_4\)-labeled
variable-to-\(X\)-type-check message is
\[
  m_{j\to i}^{X,(\ell)}(\phi(x,z))
  \propto
  Q_j(x,z)
  \prod_{k\in\partial^X_j\setminus\{i\}}
  \widehat{\nu}_{k\to j}^{(\ell)}(z)
  \prod_{k\in\partial^Z_j}
  \widehat{\mu}_{k\to j}^{(\ell)}(x).
\]
In the coupled binary graph, the message sent from \(z_j\) to the
\(X\)-type check \(i\) through the local factor \(Q_j\) is
obtained by marginalizing the same
expression over \(x\).  This proves \eqref{eq:var_message_x_equiv}.  The proof
of \eqref{eq:var_message_z_equiv} is the same with \(x\) and \(z\) interchanged.

For the check update, the \(\mathbb F_4\)-labeled message from an \(X\)-type check
\(i\) to the variable \(\alpha_j\), with \(j\in\partial_i^X\), is
\[
  \widehat{m}_{i\to j}^{X,(\ell+1)}(\phi(x,z))
  \propto
  \sum_{\substack{(\alpha_b)_{b\in\partial_i^X\setminus\{j\}}\\
  z+\sum_{b\in\partial_i^X\setminus\{j\}}
  z(\alpha_b)=s_i^Z}}
  \prod_{b\in\partial_i^X\setminus\{j\}}
  m_{b\to i}^{X,(\ell)}(\alpha_b).
\]
Write \(\alpha_b=\phi(x_b,z_b)\).  Since the constraint depends only on the
\(z_b\)'s, we may first sum over the \(x_b\)'s and use
\[
  \sum_{x_b\in\mathbb F_2}
  m_{b\to i}^{X,(\ell)}(\phi(x_b,z_b))
  \propto
  \nu_{b\to i}^{(\ell)}(z_b).
\]
The result is exactly the binary sum-product check update for the \(X\)-type row,
which proves \eqref{eq:check_message_x_equiv} at the next iteration.  The
\(Z\)-type check case proves \eqref{eq:check_message_z_equiv} in the same way.
Induction from compatible initialization proves the message identities for all
iterations.

Finally, substituting \eqref{eq:check_message_x_equiv} and
\eqref{eq:check_message_z_equiv} into the four-state belief gives
\[
  b_{4,j}^{(\ell)}(\phi(x,z))
  \propto
  Q_j(x,z)
  \prod_{i\in\partial^X_j}
  \widehat{\nu}_{i\to j}^{(\ell)}(z)
  \prod_{i\in\partial^Z_j}
  \widehat{\mu}_{i\to j}^{(\ell)}(x)
  \propto
  b_{2,j}^{(\ell)}(x,z).
\]
After normalizing both sides, this is \eqref{eq:belief_equivalence}.  Therefore,
with the common tie-breaking rule above, maximum-belief hard decisions satisfy
\(\widehat{\alpha}_j^{(\ell)}
=\phi(\widetilde{x}_j^{(\ell)},\widetilde{z}_j^{(\ell)})\).
\end{proof}

The proof shows why no special property of a particular code instance is
needed.  The check constraints in the four-state graph depend on only one
binary coordinate of each label, while the local prior contains both
coordinates.  Sum-product can therefore sum out the irrelevant coordinate in a
four-state variable-to-check message and recover exactly the coupled binary
message.  This observation is what separates joint BP from separate BP: the
former keeps the joint prior in the factor graph,
whereas the latter removes it before message passing starts.

\section{Log-Likelihood-Ratio Domain Form of the Coupled Binary Decoder}
\label{sec:llr_gf2}

The joint BP decoder can be written directly in the
log-likelihood-ratio (LLR) domain.  For a binary message \(r(u)\),
\(u\in\mathbb F_2\), write
\[
  L(r)=\log\frac{r(0)}{r(1)}.
\]
Let \(L_{j\to i}^{X,(\ell)}\) and \(L_{i\to j}^{X,(\ell)}\) denote,
respectively, the LLRs of the messages
\(\nu_{j\to i}^{(\ell)}\) and \(\widehat{\nu}_{i\to j}^{(\ell)}\) on an
\(X\)-check edge.  These are messages about the \(z_j\) variable.  Similarly,
let \(L_{j\to i}^{Z,(\ell)}\) and \(L_{i\to j}^{Z,(\ell)}\) denote the LLRs of
\(\mu_{j\to i}^{(\ell)}\) and \(\widehat{\mu}_{i\to j}^{(\ell)}\) on a
\(Z\)-check edge, which are messages about the \(x_j\) variable.

At a variable/local-prior factor \(Q_j\), define the incoming check sums
\[
  A_{j\setminus i}^{X,(\ell)}
  =
  \sum_{k\in\partial^X_j\setminus\{i\}} L_{k\to j}^{X,(\ell)},
  \quad
  A_j^{X,(\ell)}
  =
  \sum_{k\in\partial^X_j} L_{k\to j}^{X,(\ell)},
\]
and define \(A_{j\setminus i}^{Z,(\ell)}\) and \(A_j^{Z,(\ell)}\) analogously.
The local \(Q_j\) update is most simply computed by temporarily converting the
incoming LLR sums back to two-entry weights.  Set
\[
  \rho_j^{X,(\ell)}(z)\propto \exp\!\left((1-z)A_j^{X,(\ell)}\right),
  \qquad
  \rho_j^{Z,(\ell)}(x)\propto \exp\!\left((1-x)A_j^{Z,(\ell)}\right).
\]
The proportionality constants in \(\rho_j^{X,(\ell)}\) and
\(\rho_j^{Z,(\ell)}\) may be chosen arbitrarily, for example to avoid overflow,
because they cancel when the result is converted back to an LLR.  The local
marginalizations through the \(2\times2\) table \(Q_j\) are
\[
  \eta_j^{X,(\ell)}(z)
  =
  \sum_{x\in\mathbb F_2} Q_j(x,z)\rho_j^{Z,(\ell)}(x),
  \qquad
  \eta_j^{Z,(\ell)}(x)
  =
  \sum_{z\in\mathbb F_2} Q_j(x,z)\rho_j^{X,(\ell)}(z).
\]
Then the variable-to-check updates are
\begin{align}
  L_{j\to i}^{X,(\ell)}
  &=
  A_{j\setminus i}^{X,(\ell)}
  +
  \log
  \frac{\eta_j^{X,(\ell)}(0)}
       {\eta_j^{X,(\ell)}(1)},
  \qquad i\in\partial^X_j,
  \label{eq:llr_var_x}\\
  L_{j\to i}^{Z,(\ell)}
  &=
  A_{j\setminus i}^{Z,(\ell)}
  +
  \log
  \frac{\eta_j^{Z,(\ell)}(0)}
       {\eta_j^{Z,(\ell)}(1)},
  \qquad i\in\partial^Z_j.
  \label{eq:llr_var_z}
\end{align}
The check updates are the usual syndrome-modified binary parity-check updates:
\begin{align}
  L_{i\to j}^{X,(\ell+1)}
  &=
  (-1)^{s_i^Z}\,
  2\operatorname{atanh}\!\left(
  \prod_{b\in\partial_i^X\setminus\{j\}}
  \tanh\!\left(\frac{L_{b\to i}^{X,(\ell)}}{2}\right)\right),
  \qquad j\in\partial_i^X,
  \label{eq:llr_check_x}\\
  L_{i\to j}^{Z,(\ell+1)}
  &=
  (-1)^{s_i^X}\,
  2\operatorname{atanh}\!\left(
  \prod_{b\in\partial_i^Z\setminus\{j\}}
  \tanh\!\left(\frac{L_{b\to i}^{Z,(\ell)}}{2}\right)\right),
  \qquad j\in\partial_i^Z.
  \label{eq:llr_check_z}
\end{align}
Equivalently, the check update may be written with the standard box-plus
operator, and it may be approximated by the min-sum rule in the usual way.
Such approximations are useful binary-message implementations, but the theorem
above concerns the exact sum-product update.  A min-sum, damped, or clipped
implementation is exactly equivalent to a four-state implementation only if the
same approximation is applied after the same relabeling and marginalization.
The posterior joint belief at qubit \(j\) is represented by the unnormalized
log-belief
\[
  B_j^{(\ell)}(x,z)
  =
  \log Q_j(x,z)
  +(1-z)A_j^{X,(\ell)}
  +(1-x)A_j^{Z,(\ell)}
  +\mathrm{const.}
\]
The joint binary hard decision is therefore
\[
  (\widetilde{x}_j^{(\ell)},\widetilde{z}_j^{(\ell)})
  =
  \operatorname*{arg\,max}_{(x,z)\in\mathbb F_2^2}
  B_j^{(\ell)}(x,z).
\]
Since \(B_j^{(\ell)}(x,z)\) is the logarithm of the unnormalized belief
\(b_{2,j}^{(\ell)}(x,z)\), and
\eqref{eq:belief_equivalence} gives
\(b_{4,j}^{(\ell)}(\phi(x,z))=b_{2,j}^{(\ell)}(x,z)\), the \(\mathbb F_4\)
estimate defined by maximizing \(b_{4,j}^{(\ell)}\) satisfies
\[
  \widehat{\alpha}_j^{(\ell)}
  =
  \phi\!\left(\widetilde{x}_j^{(\ell)},\widetilde{z}_j^{(\ell)}\right).
\]
No probability-domain normalization constants are needed in this LLR
form; all such constants have cancelled in the log-ratios.  This display assumes
the entries of \(Q_j\) used inside the logarithm are positive; zero-probability
entries are treated in the probability domain, or by the corresponding
log-domain limits.

\section{Conclusion}

For CSS syndrome decoding, the posterior can be written as a coupled binary
factor graph.  The check factors are binary parity checks, and the factor
\(Q_j(x_j,z_j)\) is the only factor that contains both Pauli components at
qubit \(j\).  Sum-product on this factorization gives a binary BP decoder that
uses the local \(X/Z\) channel correlation whenever that correlation is present
in \(Q_j\).  If \(Q_j\) is replaced by its two marginals, the resulting decoder
is the separate BP baseline and the local correlation is absent from
the factorization.

The \(\mathbb F_4\)-labeled four-state factor graph is obtained from the same
CSS posterior by the relabeling \(\alpha=x+\omega z\).  Here
\(\mathbb F_4\) is used as a convenient label set for the four local Pauli
states in the decoding factor graph, not as a claim that a CSS code should be
viewed primarily as a quaternary code construction.  Sum-product on that graph
and sum-product on the coupled binary graph produce corresponding messages and
beliefs after relabeling and marginalization.  In this precise sense, four-state
BP and joint BP are two representations of the same probability-domain sum-product
computation for CSS syndrome decoding.
Thus the natural decoder representation depends on the code class: four-state
BP is a natural language for general stabilizer codes, while joint BP is the natural
CSS-specialized language because it keeps the binary Tanner-graph structure and
still retains the local \(X/Z\) correlation.

\bibliographystyle{IEEEtran}
\bibliography{irregular_joint_bp_refs}

\end{document}